\documentclass[11pt]{amsart}
\usepackage{amssymb}
\usepackage{amsaddr}
\usepackage[T1]{fontenc}   
\usepackage{xfrac}
\usepackage{tikz} 
\usetikzlibrary{automata,positioning}
\usepackage{amsmath}
\usepackage{algorithm}
\usepackage[noend]{algpseudocode}
\usepackage{listings}

\usepackage[utf8]{inputenc}
\usepackage{lmodern}
\usepackage{enumitem}
\newtheorem{thrm}{Theorem}

\newtheorem{prop}{Proposition}
\newtheorem{prob}{Problem}
\theoremstyle{definition}
\newtheorem{df}{Definition}

\theoremstyle{plain}

\makeatletter
\def\BState{\State\hskip-\ALG@thistlm}
\makeatother

\author[A. Polański]{Artur Polański}
\address{Institute of Computer Science and Computer Mathematics
\\Jagiellonian University in Cracow; Krak\'ow, Poland
}
\email{artur.polanski@uj.edu.pl}

\author[E. Lipka]{Eryk Lipka}
\address{Institute of Mathematics 
\\Pedagogical University of Cracow;  Krak\'ow, Poland
}
\email{eryklipka0@gmail.com}
\thanks{Research of the second author was supported by a grant of the National Science Centre, Poland, no. UMO-2019/34/E/ST1/00094}

\title[Detecting a single fault in a DFA]{Detecting a single fault in a deterministic finite automaton}
\begin{document}
\begin{abstract}
Given a DFA and its implementation with at most one single fault, that we can test on a set of inputs, we provide an algorithm to find a test set that guarantees finding whether the fault exists.
\end{abstract}

\maketitle

\section{Introduction}

In 2017 A. Roman considered in \cite{roman} the problem of finding a possible fault in an implementation of a DFA, providing a heuristic algorithm for finding a small test set recognizing whether the fault indeed exists. Unfortunately, that paper has a few errors and imprecisely described reasoning. We aim to correct those, as well as provide both an exponential time algorithm for finding the smallest possible test set, as well as a polynomial time heuristic algorithm that usually finds smaller test sets than one presented in \cite{roman}.

\section{Notation and terminology}

\begin{itemize}
\item For an alphabet $\Sigma$ we denote the set of words over $\Sigma$ by $\Sigma^*$.
\item We denote an empty word by $\varepsilon$.
\item We call a set o words a \textit{language}.
\end{itemize}

\begin{df}
A 5-tuple $\mathcal{A}=(Q,\Sigma,\delta,q_0,F)$ is called a \textit{deterministic finite automaton} (DFA for short) iff it consists of
\begin{itemize}
\item a nonempty finite set of \textit{states} $Q$,
\item a nonempty finite set of input letters $\Sigma$ called the \textit{alphabet},
\item a \textit{transition function} $\delta : Q \times \Sigma \rightarrow Q$,
\item an \textit{initial state} $q_0 \in Q$,
\item a set of \textit{final} (\textit{accept}, \textit{terminal}) states $F \subset Q$.
\end{itemize}
The function $\delta$ can be extended on the set of words $\Sigma^*$, abusing the notation we use $\delta$ for the extension.

For a word $w$ consisting of a single letter we use the definition above. For $q \in Q$ we define $\delta(q,\varepsilon)=q$. For a word $w \in \Sigma^*$ of the form $w=au$, where $a\in \Sigma$, $u \in \Sigma^*$ we define $\delta(q,w)=\delta(\delta(q,a),u)$.

We say that a word $w \in \Sigma^*$ is \textit{accepted} (by $\mathcal{A}$) iff $\delta(q_0,w)\in F$, otherwise we say that it is \textit{rejected} (by $\mathcal{A}$).

An automaton  $\mathcal{A}$ is called \textit{minimal} iff there is no automaton $$\mathcal{A}'=(Q',\Sigma',\delta',q_0',F')$$ with $|Q'|<|Q|$ such that the sets of words accepted by $\mathcal{A}$ and $\mathcal{A}'$ are equal. For any automaton there exists, up to isomorphism, exactly one minimal automaton which accepts the same set of words (see e.g. \cite{hopcroft}).

For $Q'\subset Q$ and $q_1\in Q$ we say that a word $w \in \Sigma^*$ is \textit{synchronizing} set $Q'$ to state $q_1$ iff $\forall_{q\in Q'} \delta(q,w)=q_1$.

We call a state $q \in Q$ a \textit{sink state} (or a \textit{sinkhole state}) iff for any $a \in \Sigma$ we have $\delta(q,a)=q$. If a sinkhole $q$ is a final state we call it a \textit{positive sinkhole} and we call it a \textit{negative sinkhole} otherwise.
\end{df}

When talking about automatons it is often convenient to use their graph representation, below there is an example of an automaton $$\mathcal{A}=(\{1,2,3,\text{A},\text{X}\},\{a,b\},\delta,1,\{A\}).$$ The function $\delta$ is depicted by arrows, the initial state is the only one with an arrow pointing to it that is not a transition, and a double circle around a node represents final states. That exact automaton is used in \cite{roman}, which we will also use as an example later.

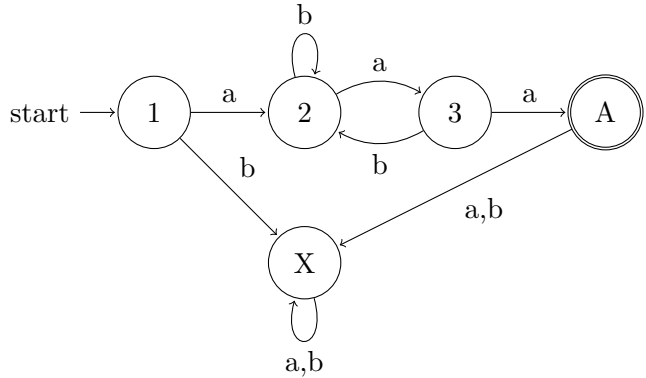
\begin{figure}[h]
    \centering
\begin{tikzpicture}[shorten >=1pt,node distance=2cm,on grid,auto] 

   \node[state,initial] (q_1)   {1}; 
   \node[state] (q_2) [right=of q_1] {2}; 
   \node[state] (q_3) [right=of q_2] {3}; 
   \node[state,accepting](q_4) [right=of q_3] {A};
   \node[state] (q_5) [below=of q_2] {X}; 
	\path[->] 
    (q_1) edge  node {a} (q_2)
    (q_2) edge  [bend left] node {a} (q_3)
          edge [loop above] node {b} ()
    (q_3) edge  node {a} (q_4)
          edge [bend left] node {b} (q_2)
    (q_1) edge  node {b} (q_5)
    (q_4) edge  node {a,b} (q_5)
    (q_5) edge [loop below] node {a,b} ();
\end{tikzpicture}

    \caption{Example automaton from \cite{roman}}
    \label{picture_roman}
\end{figure}

\section{preliminaries}

Firstly, we define what we mean by a fault in an automaton. Let $\mathcal{A}=(Q,\Sigma,\delta,q_0,F)$ be a DFA, and assume that its implementation (SUT - system under testing) represents $\mathcal{A}'=(Q,\Sigma,\delta',q_0,F)$, a DFA that may be different. When thinking about possible differences between $\mathcal A$ and $\mathcal A'$, there are a few natural candidates, for example:
\begin{itemize}
\item an incorrect transition, i.e. $\delta(q,a) \neq \delta'(q,a)$ for some $q\in Q$, $a \in \Sigma$,
\item a missing transition (first possible interpretation), i.e. after $\delta'(q,a)$ is used for some $q\in Q$, $a \in \Sigma$, the word being read will be rejected,
\item a missing transition (second possible interpretation), i.e. after $\delta'(q,a)$ is used, the state will not change, when in $\mathcal{A}$ it does; this is covered by our first candidate, an incorrect transition.
\end{itemize}
If we ease the requirements, so not only transition functions differ we can also consider:
\begin{itemize}
    \item incorrect initial state,
    \item incorrect set of accepting states,
    \item incorrect set of states;
\end{itemize}
but such cases are out of the scope of this paper.

Consider an incorrect transition first. If $\mathcal{A}$ is not minimal, we may not be able to conclude whether the fault occurred. For example, if $F$ is empty (or $F=Q$) and $|Q|>1$, then all words are rejected (or accepted), an incorrect transition (or, for that matter, any number of incorrect transitions) will not change that. On the other hand, we have the following.

\begin{prop}\label{propMin}
A single incorrect transition in a minimal DFA can always be detected.
\end{prop}

\begin{proof}
Let $\mathcal{A}=(Q,\Sigma,\delta,q_0,F)$ be a DFA, and $\mathcal{A'}=(Q,\Sigma,\delta',q_0,F)$ be its implementation with a single incorrect transition. Assume that that fault cannot be detected. There exist $q,q_1, q_2 \in Q$, $a \in \Sigma$ such that $\delta(q,a)=q_1$ and $\delta'(q,a)=q_2$, where $q_1 \neq q_2$. Since the fault cannot be detected, then $\{w\in \Sigma^* : \delta(q_1,w)\in F\}=\{w\in \Sigma^* : \delta'(q_2,w)\in F\}$. It follows that if we remove $q_2$ from $Q$ and replace all transitions to $q_2$ with transitions to $q_1$ (possibly changing the initial state to $q_1$, if $q_2$ was initial), we will get an automaton with fewer states that accepts the same language as $\mathcal{A}$, a contradiction with minimality.
\end{proof}

That is the reason why, from now on, we assume that our DFA is minimal. In \cite{roman} we read that the first interpretation of a missing transition can also be realised by an incorrect transition. That is not always the case however. Imagine a DFA with only one state that is final. An incorrect transition in that case is not even possible, but a missing transition (first interpretation) is. To be more precise, such realisation is possible iff $\mathcal A$ has a negative sinkhole.

Going forward, we assume that the only possible \textit{fault} is an incorrect transition.

\section{finding the minimal test set}
Let us state our goal as a decision problem.
\begin{prob} \label{main_problem}
For a given integer $k$, a DFA $\mathcal{A}$ and its implementation $\mathcal{B}$ that is either the same as $\mathcal{A}$ or differs from $\mathcal{A}$ by a single fault, verify the existence of a set of at most $k$ words that can detect is $\mathcal{B}$ differs from $\mathcal{A}$.
\end{prob}

For this part we need to know that each DFA corresponds (one-to-one) with a regular expression that represents the language accepted by that DFA.

\begin{thrm}[Kleene, \cite{kleene}]
A language is regular if and only if there exists DFA, which accepts this language.
\end{thrm}

Consider $\mathcal{A}=(Q,\Sigma,\delta,q_0,F)$, a DFA. For every $a\in \Sigma$ and $q, q' \in Q$ such that $\delta(q,a) \neq q'$ we can construct an automaton with a single fault $\mathcal{A}_{(q,a,q')}=(Q,\Sigma,\delta_{(q,a,q')},q_0,F)$ such that $\delta_{(q,a,q')}(q,a)=q'$ is its only fault. For any such automaton we can consider $$\widehat{\mathcal{A}_{(q,a,q')}}=(Q\times Q,\Sigma,\widehat{\delta_{(q,a,q')}},(q_0,q_0),F\times (Q \setminus F) \cup (Q \setminus F)\times F) ,$$
where $\widehat{\delta_{(q,a,q')}}((q_1,q_2),w) = (\delta(q_1,w),\delta_{(q,a,q')}(q_2,w))$, which accepts a non\-em\-pty language (by Proposition \ref{propMin}) of words that recognize the fault $$\delta_{(q,a,q')}(q,a)=q',$$ because these are precisely the words that are accepted by $\mathcal{A}$ and not $\mathcal{A}'$, or the other way around. That way, out problem can be restated as finding a minimal set of words for a given finite set of regular expressions (i.e. languages accepted by the automatons $\widehat{\mathcal{A}_{(q,a,q')}}$), such that for each given regular expression at least one word from our set satisfies it. Unfortunately, we have the following. 

\begin{prop}
Given a finite set of regular expressions $R$ and a positive integer $k$, the problem if there exists a set of words $W$, such that $|W| \leqslant k$ and a language represented by each regular expression in $R$ has a nonempty intersection with $W$ is NP-complete.
\end{prop}

\begin{proof}
If we are provided with a set $W$ of at most $k$ words it is easy to check in polynomial time whether each regular expression in $R$ represents a language that has at least one word in $W$, therefore the problem is in NP.

To show that it is NP-complete, we will reduce a dominating set problem to it. Recall that a dominating set of a graph is a subset $A$ of its vertices such that any vertex is either in $A$ or has a neighbour in $A$. Let $G=(V,E)$ be a graph such that $V=\{v_1,\ldots,v_n\}$. We construct a set of regular expressions as follows. For $i \in \{1,\ldots ,n\}$ let $r_i$ be a regular expression over the alphabet $V$ such that it is a sum of $v_i$ and all its neighbours in $G$. We need to show that a vertex cover $V' \subset V$ such that $|V'| \leqslant k$ exists iff there exists a set of words $W'$ such that $|W'|\leqslant k$ and each regular expression from $\{r_1,\ldots,r_n\}$ is represented.

Indeed, if we can represent all regular expressions ${r_1,\ldots,r_n}$ with a set of words $|W'|$ such that $|W'|\leqslant k$ (note that each word must be a single letter $v_i$), then a corresponding set $|V'|$ of vertices with labels corresponding to $W'$ is a vertex cover and $|V'|=|W'|$.

On the other hand, given a vertex cover $V'$, if we take $W'$ as letters $v_i$ such that a vertex with that label is in $V'$ we get a set of words as needed.

As dominating set problem is NP-complete (problem GT2 in \cite{garey}), that concludes the proof.
\end{proof}

That means of course, that we can find a smallest set of words that recognize a fault in DFA, but unless P=NP, proceeding as described above will not produce the needed set in polynomial time. We propose the following.

\begin{prob}
Decide whether problem \ref{main_problem} is NP-complete.
\end{prob}

\section{heuristic algorithm}
\subsection{Algorithm proposed by Roman}

The algorithm presented in \cite{roman} is the following:
\begin{algorithm}
\caption{ GENERATE from \cite{roman} \protect\\ 
\protect\textbf{Input}: a finite automaton $\mathcal A = (Q,\Sigma,\delta_M,q_0,F)$ \protect\\
\protect\textbf{Output}: a set $T \subset \Sigma^*$ of test cases that detect all possible "incorrect transition" faults }
\label{algorithm_roman}
\begin{algorithmic}[1]
\State $E = \{(q,a) \in Q \times \Sigma\},\ T = \varnothing$
\While {$E \neq \varnothing$}
\State \textbf{find} a path $p$ such that $\delta_M(q_0,\sigma (p)) \in F$ and which maximizes $|e(p) \cap E|$ and for which $e(p) \cap E \neq \varnothing$.
\If {such $p$ exists}
\For {$(r,a) \in e(p)$}
\State $E = E \setminus \{(r,a)\}$
\State \textbf{let} $\sigma(p)=uw$, where $\delta_M(q_0,u)=r$
\ForAll{$q \in Q \setminus \{\delta_M (r,a)\}$} 
\State $\delta_M' = \delta_M$, $\delta_M'(r,a)=q$
\If {$\delta_M'(q_0,\sigma(p)) \in Q \setminus F$}
\State $T = T \cup \sigma(p)$
\Else
\State \textbf{find} $w' \in \Sigma^*$ and $b \in \Sigma \setminus \{a\}$ such that only one of states $\delta_M(q_0,ubw'), \delta_M'(q_0,ubw')$ is a terminal state.
\State $T = T \cup \{ubw'\}$
\EndIf
\EndFor
\EndFor
\ElsIf{$E \neq \varnothing$}
\State \textbf{find} a path $p$ that covers as many elements from $E$ as possible (ending in $Q \setminus F$)
\ForAll{$(r,a) \in e(p)$}
\State $E = E \setminus \{(r,a)\}, S= Q \setminus \delta(r,a)$
\State \textbf{let} $\sigma(p)=uw$, where $\delta_M(q_0,u)=r$
\While{$S \neq \varnothing$}
\State \textbf{find} a word $w$ that synchronizes to some accepting state from $F$ as many states from $Q\setminus \delta(r,a)$ as possible and remove them from $S$
\State $T = T \cup \{uaw'\}$
\EndWhile
\EndFor
\EndIf
\EndWhile
\State \Return T
\end{algorithmic}
\end{algorithm}

As for the terminology used there (following \cite{roman}), a path $p$ is a sequence of states $p=(q_1,q_2,\ldots,q_k)$ such that for each $i \in \{1,2,\ldots,k-1\}$ there exists $a_i \in \Sigma$ such that $\delta_M(q_i,a_i)=q_{i+1}$ (it corresponds of course to a directed path in a transition graph for $\mathcal{A}$). The letters $a_i$ form a word $a_1 \ldots a_{k-1}$ which is denoted by $\sigma(p)$ and a sequence of pairs $(p_1,a_1) \ldots (p_{k-1},a_{k-1})$ is denoted by $e(p)$. We will later use the same notation in our algorithm.

It will be helpful to summarize how the algorithm works, as later we will do the same for our version, the conceptual differences should be easier to see that way. We first find a path through the transition graph ending in a final state that covers the most transitions possible. We then explore all possible faults that may cause a deviation in the path taken through the transition graph with the same input word. If the deviation causes the automaton to end in a state that is not final, then the original word is "good" for finding that particular fault. If not (i.e. even though there was a deviation, we still ended up in a final state), then our word is not "good" for that particular fault, and we just find a word that is. After that we can repeat the procedure, but considering all the transitions our word passed through as "done". That way we can eventually label on transitions that can be used as part of a path from an initial to a final state. That, of course, may still not be the end, since all the words we find take us to a final state, there can still be transitions left that are not covered by the procedure thus far (for example, those leading to a negative sinkhole). We therefore find a path again covering the most transitions not yet "done" (necessarily we end in a non-final state). Now we have a similar situation to the one before, either a single deviation will send us to a final state (in which case that one is covered by the word we are in process of analyzing), or not. Here however, when considering a certain state and possible faults that have it as a starting point, that result in paths leading to a non-final state, we attempt to lower the number of words using synchronization. To be more precise, if considering a path $p$ through a state $q$ we find that changing a single transition can result in multiple paths ending in a non-final state, instead of adding a word for each case we use synchronization to find one word that synchronizes as many states that can be reached just after the fault as possible into one final state.

Unfortunately, there are a few problems with this algorithm.
\begin{enumerate}
\item In line 13 of the algorithm $ubw'$ should be changed to $uaw'$, $b$ is not needed, indeed as written, algorithm may not work.
\item In line 21, should be noted what automaton that synchronization applies to (what about the incorrect transition? we are synchronizing across multiple possibilities for it).
\item There is no exact procedure for "find a path" in 3. Even though there can be more than one path covering a maximal number of transitions, the final result may depend on which is chosen.
\item If we aim to deliver the smallest possible set of words, why use the optimization in only one of two cases? It is easy to see that a minimal automaton remains minimal after changing all non-final states into final and vice-versa. It is therefore peculiar that when considering paths that end in a final state we do not optimize by synchronization as in the case of paths that end in a non-final state, even though they are, in a sense, symmetric.
\end{enumerate}

There are unfortunately also problems in the example shown in \cite{roman} (using the algorithm on a concrete automaton, the one from figure \ref{picture_roman}). The words are wrong (probably because $\delta_M$ and $\delta_M'$ got mixed up). If we assume that the synchronization is done in such a way that the incorrect transition is not used, then the final set achieved by this approach may be for example
\begin{align*}
    \{ababaa,aa,bbaa,baaa,&b,baaaa,babaa,ba,babaaa,babbaa,bab, \\ &aaaaaaa,aaaabaa,aaaa,aaabaaa,aaabaa,aaab\}.
\end{align*}
What set we obtain does depends on the approach used in "find" statements (as we suggested earlier). For example, if the first word we try is $aabbaa$ instead of $ababaa$ we no longer need $aa$.

We will now proceed to describe possible improvements to the algorithm.

\subsection{Improvements to the algorithm (without synchronization)} As an improvement to algorithm \ref{algorithm_roman} we propose algorithm \ref{algorithm_ours}. We will describe how it works just like we did for algorithm \ref{algorithm_roman}.

\begin{algorithm}
\caption{ Generating test cases (without synchronization)
\protect\\ 
\protect\textbf{Input}: a finite automaton $\mathcal A = (Q,\Sigma,\delta_M,q_0,F)$ \protect\\
\protect\textbf{Output}: a set $T \subset \Sigma^*$ of test cases that detect all possible "incorrect transition" faults }
\label{algorithm_ours}
\begin{algorithmic}[1]
\State $E = \{(q,a) \in Q \times \Sigma\},\ T = \varnothing$
\While {$E \neq \varnothing$}
\State \textbf{find} a path $p$ such that $\delta_M(q_0,\sigma (p)) \in F$ and which maximizes $|e(p) \cap E|$ and for which $e(p) \cap E \neq \varnothing$.
\If {such $p$ exists}
\State $T_{acc} = \varnothing, T_{rej} = \varnothing$
\For {$(r,a) \in e(p)$, from the last edge of $p$ to the first}
\State $E = E \setminus \{(r,a)\}$
\State \textbf{let} $\sigma(p)=uw$, where $\delta_M(q_0,u)=r$
\For {$q \in Q \setminus \{\delta_M (r,a)\}$} 
\State $\delta_M' = \delta_M$, $\delta_M'(r,a)=q$
\If {$\delta_M'(q_0,\sigma(p)) \in Q \setminus F$}
\State $T = T \cup \sigma(p)$
\Else
\If {$\delta_M'(q_0,T_{acc}) \cap Q \setminus F \neq \varnothing$ \textbf{or} $\delta_M'(q_0,T_{rej}) \cap F \neq \varnothing$} do nothing
\Else
\If {there exists $w' \in \Sigma^*$ such that $\delta_M(q_0,uaw') \in F$ and $\delta_M'(q_0,uaw') \in Q \setminus F$}
\State $T_{acc} = T_{acc} \cup \{uaw'\}$
\Else
\If {there exists $w' \in \Sigma^*$ such that $\delta_M(q_0,uaw') \in Q \setminus F$ and $\delta_M'(q_0,uaw') \in F$}
\State $T_{rej} = T_{rej} \cup \{uaw'\}$
\Else
\State $E = E \cup \{(r,a)\}$ and consider the edge again when it appears in the list of edges of $p$
\EndIf
\EndIf
\EndIf
\EndIf
\EndFor
\EndFor
\State $T = T \cup T_{acc} \cup T_{rej}$
\Else
\If {$E \neq \varnothing$}
\State \textbf{find} a path $p$ that covers as many elements from $E$ as possible (ending in $Q \setminus F$) and continue similarly to the previous case
\EndIf
\EndIf 
\EndWhile
\State \Return $T$
\end{algorithmic}
\end{algorithm}

The very beginning is the same, we start with a path through the transition graph covering the maximal possible number of transitions, that end in a final state, let us denote the corresponding word by $w$. Since what we want to minimize is the number of words in the final test set, we then proceed from the longest to the shortest prefix of $w$ when considering the possible faults that may affect $w$. The reason is, if we find a word that is accepted in $\mathcal{A}$, but rejected when a given fault occurs after a long prefix, it may also work when considering possible faults occurring after shorter prefixes.

Note that in line 22 we are indeed able to consider the edge again later, because we can only reach this instruction if the transition occured more than once in $p$.

After we exhaust words that are accepted by $\mathcal{A}$, since we are left with a case that is essentially symmetric to the one just done (when in a minimal automaton we change all final states to non-final and vice versa, the resulting automaton is still minimal), we can proceed in a way mirroring the "accepted words" case. In particular that means that we can try to use this algorithm on the "complemented" DFA and check whether it yields a smaller set of test cases (of course we also need to swap acceptance and rejection for the implementation, but that can be done after the implementation actually gives us a result).

That also means that we can use synchronization to try limiting the size of our final set further in both "accepting" and "rejecting" cases, not just in one like in algorithm \ref{algorithm_roman}.

\subsection{Improvements to the algorithm (with synchronization)}

When we consider a transition and where it might lead if the fault replacing this exact one, some options are already taken care of by the word chosen for this transition (like we discussed before, we may reach a final state where the original word was rejected or vice-versa). Instead of choosing a word for each of the other options, we can try to "synchronize" to a rejecting state if the word chosen for the transition is accepted, or to an accepting state otherwise. In \cite{roman} it was done only in a latter case. It was not stated explicitly in \cite{roman} but can be presumed that the synchronization was done on the automaton with the transition considered removed, starting from a set of states that were all options for where the fault may lead, with the ones that were already covered by the word omitted.

It can be done better, however. Consider a word of the form $ua$, where $a \in \Sigma$ is the letter that in $\mathcal{A}$ corresponds to the transition we have yet to cover in an accepted word $uaw$. For all the options not covered by the words so far, we can try to find a small set of words $W'$ such that members of the set $\{uaw' : w' \in W'\}$ cover all of those options. We can do this by constructing regular expressions for each of the options that reject the word if started there and try to find a smallest possible set of words such that each of the expressions has a representative. Since we do not want to find a minimal set for this problem (exponential complexity), we can, for example, do it greedily. Synchronization can be used here as a preliminary, because if a set of states can be synchronized, the language of words satisfying all of their corresponding regular expressions is nonempty. 

\subsection{Using synchronization to avoid masking}
In both algorithms we find paths $p$ which maximize $|e(p)\cap E|$. There can, of course, be more then one such path. What is more, when considering two paths, using one of them can be better for our purposes then using the other. 

For a path $p$ in the transition graph of $\mathcal{A}$, and an edge $x\xrightarrow{a}y$ in $p$, we say that $p$ is \textit{masking} a faulty transition $x\xrightarrow{a}z$ if $y \neq z$ and replacing $x\xrightarrow{a}y$ with $x\xrightarrow{a}z$ in $\delta_M$ does not change $\sigma(p)$ being accepted by $\mathcal{A}$ (i.e. it is accepted both before and after or rejected both before and after introducing the fault).

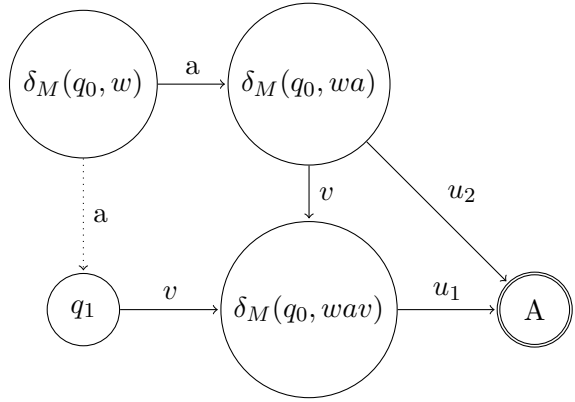
\begin{figure}[h]
    \centering
    \begin{tikzpicture}[shorten >=1pt,node distance=3cm,on grid,auto] 
   
   \node[state] (q_2) {$\delta_M(q_0, w)$}; 
   \node[state] (q_3) [right=of q_2] {$\delta_M(q_0, wa)$}; 
   \node[state](q_4) [below=of q_3] {$\delta_M(q_0, wav)$};
   \node[state] (q_5) [below=of q_2] {$q_1$}; 
   \node[state,accepting] (q_1) [right=of q_4] {A};
	\path[->] 
    (q_2) edge  node {a} (q_3)
    (q_3) edge  node {$v$} (q_4)
    (q_5) edge  node {$v$} (q_4)
    (q_4) edge node {$u_1$} (q_1)
    (q_3) edge node {$u_2$} (q_1);
    \path[->,dotted]
    (q_2) edge node {a} (q_5);
\end{tikzpicture}

    \caption{Example of masking}
    \label{picture_masking}
\end{figure}

It is unrealistic to try and find a path $p$ that maximizes $|e(p)\cap E|$ that does not mask anything, as it would imply that we can cover all possible faulty transitions that disrupt $p$ just by $\sigma(p)$, but we can try to avoid some easy to find maskings.

Consider the example shown in figure \ref{picture_masking}. In the situation shown there, $ \delta_M(q_0, wa)$ and $q_1$ can be synchronized using the word $v$ (suppose that the path $q_1\xrightarrow{v}\delta_M(q_1,v)$ does not use the edge $\delta_M(q_0,w)\xrightarrow{a} \delta_M(q_0,wa)$). In that case any path with a prefix $wav$ is masking, so we may want to avoid it (as we said, it may not be possible to avoid every possible masking). Before generating test cases we can find which pairs of states can be synchronized in $\mathcal{A}$ and which words are used to synchronize them. This can be done by constructing a subautomaton $\mathcal P^{[2]} (A)$ of the power automaton (eg. we take 2- and 1-element subsets of $Q$ as states) and then finding paths from two element states to singletons. For each state $q_1$ we can store a regular expression describing words $w$, which synchronize $q_1$ with some other state. When choosing among paths that maximize $|e(p) \cap E|$ we then prioritize those that are not masking according to what we found.

In Roman's example (figure \ref{picture_roman}) in first iteration of algorithm we can choose two different possible paths $p_1,p_2$ which maximize $|e(p) \cup E|$, namely $\sigma(p_1) = ababaa, \sigma(p_2) = aabbaa$. By looking at the automaton in figure \ref{picture_power} we can easily see that
$$\delta_M(2,b) = \delta_M(3,b) = 2.$$ This is why we do not want to use path with prefix $ab$, as it would mask the faulty transition $1\xrightarrow{a}3$. Hence we should use $p_2$ instead of $p_1$ to avoid the need for another test case recognizing this fault like $aa$ or $aaa$.

We also have $\delta_M(1,ab) = \delta_M(2,ab) = 2$, so we could be masking some transition going into state $1$ or $2$. There are three such transitions, namely $1\xrightarrow{a}2,2\xrightarrow{b}2,3\xrightarrow{b}2$, of which the first two are used in synchronizing 1 and 2 by $ab$. The faulty transition $2\xrightarrow{b}1$ will not be masked by $p_2$, but will be by $p_1$. The fault $1\xrightarrow{a}1$ will make the SUT trivially reject everything. Hence, the only possible error that could be masked because of the synchronization by $ab$ is $3\xrightarrow{b}1$ and this would happen only if our path follows $bab$ from state 3, like in $abababaa$. Again it is better to take $aabbaa$, which also maximizes the number of used edges and does not mask.

We can have a similar situation when considered word rejected, for instance in the example automaton there is a synchronization to the negative sinkhole state ($\delta_M(Q, aaaa) = X)$, but again, if a method words for accepted words, it can be trivially modified to work for rejected ones.

\subsection{Transitions to a sinkhole}
If our automaton has a sinkhole, then it is easy to see that we must have at least one word for each transition leading to it in out test set. Even worse, for each possible fault of each of those transitions we must ensure that we have at least one word that is accepted (for a negative sinkhole) or rejected (for a positive sinkhole) by $\mathcal{A}$. We will describe the procedure for the negative sinkhole, but again, we can easily adapt it for a positive one.

Let $Q = \{1,2,\ldots, n, X\}$ where $X$ is the negative sinkhole. For any given state $q\in Q\setminus \{X\}$ and $x\in \Sigma$ such that $\delta_M(q,x)=X$ we need to test if this transition holds in SUT. To do this, for any $r \in Q\setminus \{X\}$ we can take words $u,w$ satisfying $\delta_M(q_0,u) = q, \delta_M(r,w)\in F$ and then $uxw$ will test for faulty transition $q\xrightarrow{x}r$. This way we generate $n$ test cases for every existing transition from non-sinkhole to sinkhole, but there is a way to use fewer words.

First we find a partition of $Q\setminus \{X\}$ into subsets $S_1,\ldots, S_m$ such that each of them can be synchronized to an accepting state; let $l_i$ be a synchronizing word for $S_i$. Now, for each $q\in Q \setminus \{X\}$ and $x \in \Sigma$ we take a word $w$ with $\delta_M(q_0,w)=q$ and add test cases $wxl_1, \ldots, wxl_m$.

Of course $1 \le m \le n$ and for any automaton there exists the smallest possible value of $m$. Finding this value is NP-hard, but even greedy approach gives better results than adding $n$ test cases for every non-sinkhole-to-sinkhole transition. 

It is worth noting, that there are still sinkhole-to-sinkhole loop transitions that may or may not be covered by this set of test cases. Using this as a preliminary to algorithm \ref{algorithm_ours} we get algorithm \ref{algorithm_ours2}.

\begin{algorithm}
\caption{ Generating test cases
\protect\\ 
\protect\textbf{Input}: a finite automaton $\mathcal A = (Q,\Sigma,\delta_M,q_0,F)$ \protect\\
\protect\textbf{Output}: a set $T \subset \Sigma^*$ of test cases that detect all possible "incorrect transition" faults }
\label{algorithm_ours2}
\begin{algorithmic}[1]
\State \textbf{construct} automaton $\mathcal P^{[2]} (A)$
\State \textbf{find} synchronization patterns for 2-element sets of states
\State $T_{sink} = \varnothing, E_{sink} = \varnothing$
\If{$\exists X\in Q, X$ is negative sinkhole}
\State \textbf{find} partition of $Q\setminus\{X\}$ into $S_1,\ldots, S_m$ and words $l_1, \ldots, l_m$ such, that $l_i$ synchronizes $S_i$ to an accepting state
\For{$q \in Q, x \in \Sigma : \delta(q,x)=X$}
\State \textbf{find} $w\in \Sigma^* : \delta(q_0,w)=q$
\State $T_{sink} = T_{sink} \cup \{wxl_1, wxl_2,\ldots wxl_m\}$
\State $E_{sink} = E_{sink} \cup \{(q,x)\}$
\EndFor
\EndIf
\State \textbf{run} algorithm 2 but with $E = E \setminus E_{sink}$, knowing the synchronization patterns and that some faults on sinkhole-to-sinkhole loops could already be covered by $T_{sink}$
\State \Return $T \cup T_{sink}$
\end{algorithmic}
\end{algorithm}

\subsection{Summary and an example}
To sum up, our proposal for a heuristic polynomial time solution to problem \ref{main_problem} is the algorithm \ref{algorithm_ours}, that can be further modified by adding synchronization, choosing the words more carefully and considering faults of transitions leading to sinkholes first (each described in more detail in their respective subsections). All of the presented procedures assumed that we want to minimize the final set, in practice we may of course be satisfied with a more crude approximation of the optimal solution. If not we may try to improve it and still retain polynomial time. If it is still not enough, the exponential algorithm for finding the optimal solution may always be used if the size of the automaton permits it.

Now we will show how our algorithm (with mentioned possible modifications) works on a concrete automaton, we will reuse the one from \cite{roman}. On figure \ref{picture_power} we see $\mathcal P^{[2]} (\mathcal{A})$ , it is easy to check that the only 2-element sets, that can be synchronized to an accepting state are $\{1,2\} \xrightarrow{abaa}\{A\}, \{2,3\}\xrightarrow{baa}\{A\}$. We pick a partition $S_1=\{1,2\}, S_2=\{3\}, S_3=\{A\}$ and generate test cases
$$T_{sink} = \{babaa, ba, b, aaaaabaa, aaaaa, aaaa, aaababaa, aaaba, aaab\}.$$
We then enter algorithm \ref{algorithm_ours}. In the first pass we pick the word $aabbaa$ instead of $ababaa$, because we know that it would mask something ( $ \{2,3\}\xrightarrow{a}\{2\}$). In the second pass we are only left with $E = \{(X,a), (X,b)\}$, first of them is tested by set $\{baaaa, baaa, baa, ba\}$ and second by $\{bbaaa, bbaa, bba, bb\}$. However, as mentioned before, some faults on this edges could already be covered by $T_{sink}$. In fact $babaa$ covers $\delta_M'(X,a) =2 \vee\delta_M'(X,a) =3 \vee \delta_M'(X,b) =2$ and $ba$ already is included; we only need to add words $baaaa, bbaaa,bba,bb$.

This way, we get a 14 element test set
\begin{align*}
    T = \{babaa, ba, b, aaaaabaa, &aaaaa, aaaa, aaababaa, aaaba, aaab, \\ &aabbaa,baaaa,bbaaa,bba,bb\}.
\end{align*}
In comparison to Roman's approach it saves only 3 test cases, but due to direct computation we know that minimal set for this particular automaton has at least 12 elements.

\begin{figure}
    \centering
\begin{tikzpicture}[shorten >=1pt,node distance=2cm,on grid,auto] 

   \node[state,initial] (q_1)   {1}; 
   \node[state] (q_2) [right=of q_1] {2}; 
   \node[state] (q_3) [right=of q_2] {3}; 
   \node[state,accepting](q_4) [right=of q_3] {A};
   \node[state] (q_5) [below=of q_2] {X};
   \node[state] (q_2x) [below=of q_5] {2,X};
   \node[state] (q_1x) [left=of q_2x] {1,X};
   \node[state] (q_3x) [right=of q_2x] {3,X};
   \node[state] (q_ax) [right=of q_3x] {A,X};
   \node[state] (q_2a) [below=of q_2x] {2,A};
   \node[state] (q_3a) [below=of q_ax] {3,A};
   \node[state] (q_13) [below=of q_1x] {1,3};
   \node[state] (q_23) [right=of q_ax] {2,3};
   \node[state] (q_12) [above=of q_23] {1,2};
   \node[state] (q_1a) [left=of q_13] {1,A};
	\path[->] 
    (q_1) edge  node {a} (q_2)
    (q_2) edge  [bend left] node {a} (q_3)
          edge [loop above] node {b} ()
    (q_3) edge  node {a} (q_4)
          edge [bend left] node {b} (q_2)
    (q_1) edge  node {b} (q_5)
    (q_4) edge [above] node {a,b} (q_5)
    (q_5) edge [loop left] node {a,b} ()
    (q_1x) edge node {a} (q_2x)
           edge node {b} (q_5)
    (q_2x) edge  [bend left] node {a} (q_3x)
           edge [loop above] node {b} ()
    (q_3x) edge  node {a} (q_ax)
           edge [bend left] node {b} (q_2x)
    (q_ax) edge [above] node {a,b} (q_5)
    (q_1a) edge node {a} (q_2x)
           edge [bend left = 90] node {b} (q_5)
    (q_2a) edge [below] node {a} (q_3x)
           edge node {b} (q_2x)
    (q_3a) edge node {a} (q_ax)
           edge node {b} (q_2x)
    (q_13) edge node {a} (q_2a)
           edge node {b} (q_2x)
    (q_12) edge node {a} (q_23)
           edge node {b} (q_2x)
    (q_23) edge node {a} (q_3a)
           edge node {b} (q_2);
\end{tikzpicture}

    \caption{Power automaton for Roman's example}
    \label{picture_power}
\end{figure}
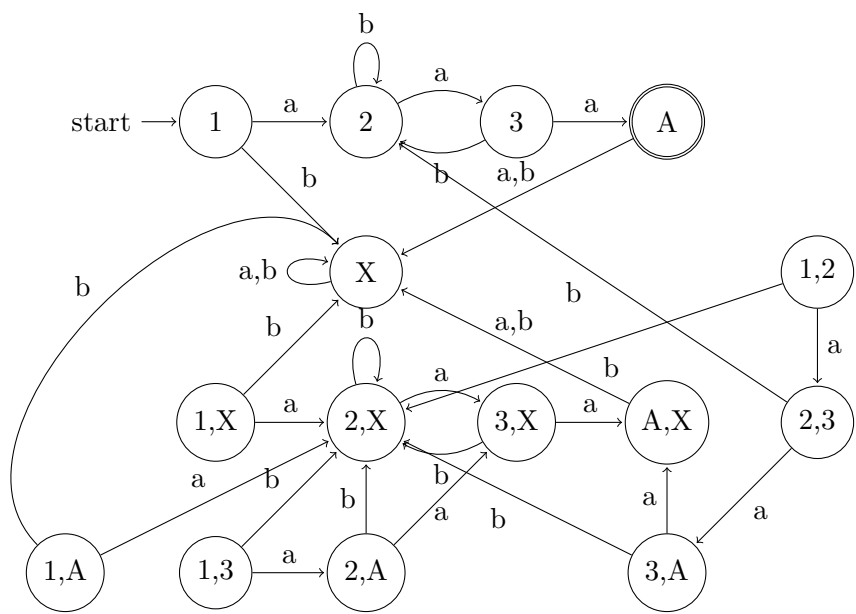

\section{Further studies}
It seems there are more possible ways to improve the heuristic algorithm. One approach may be by exhausting the set of possible faults instead of the set of edges. There is also possibility to extend the algorithm to detect other types of faults.

From theoretical point of view, the question if finding the minimal set of tests is NP-complete remains open (we suspect it is NP-complete).

\end{document}